\documentclass[runningheads,a4paper,draft]{llncs}

\usepackage{alltt}
\usepackage{amsmath}
\usepackage{amssymb}
\usepackage{nicefrac}
\usepackage{stmaryrd}
\usepackage{graphicx}
\usepackage{tikz}
\usepackage{wrapfig}
\usepackage{xspace}
\usepackage{xcolor}

\usepackage{url}
\urldef{\mailsa}\path|{benjamin.kaminski, katoen}@cs.rwth-aachen.de|    
\newcommand{\keywords}[1]{\par\addvspace\baselineskip
\noindent\keywordname\enspace\ignorespaces#1}

\begin{document}

\newcommand{\soqprogs}{\ensuremath{\soprogs_{\bbbq}}\xspace}
\newcommand{\soprogs}{\ensuremath{\textnormal{\textsf{Prog}}}\xspace}
\newcommand{\sonprogs}{\ensuremath{\textnormal{\textsf{ordProg}}}\xspace}
\newcommand{\sovars}{\ensuremath{\textnormal{\textsf{Var}}}\xspace}
\newcommand{\ttt}[1]{\textnormal{\texttt{#1}}}

\newcommand{\T}{\ensuremath{\textnormal T}\xspace}
\newcommand{\Tp}{\ensuremath{\T_p}\xspace}
\newcommand{\Tps}{\ensuremath{\Tp^*}\xspace}
\newcommand{\Exp}[2]{\ensuremath{\textnormal E_{#1}(#2)}\xspace}
\renewcommand{\Pr}{\textnormal{Pr}}

\newcommand{\toppr}{\ensuremath{\bot_{\textnormal{pr}}}}
\newcommand{\topnd}{\ensuremath{\bot_{\textnormal{nd}}}}
\newcommand{\topter}{\ensuremath{\bot_{\textnormal{ter}}}}

\newcommand{\Problem}[1]{\mathcal{#1}}
\newcommand{\cProblem}[1]{\overline{\text{\footnotesize{$\Problem{#1}$}}}}

\newcommand{\HP}{\Problem{H}\xspace}
\newcommand{\cHP}{\cProblem{H}\xspace}
\newcommand{\UHP}{\Problem{UH}\xspace}
\newcommand{\cUHP}{\cProblem{UH}\xspace}
\newcommand{\COF}{\Problem{COF}\xspace}
\newcommand{\cCOF}{\cProblem{COF}\xspace}
\newcommand{\AST}{\ensuremath{\Problem{AST}}\xspace}
\newcommand{\cAST}{\ensuremath{\cProblem{AST}}\xspace}
\newcommand{\UAST}{\ensuremath{\Problem{U\hspace{-1pt}AST}}\xspace}
\newcommand{\cUAST}{\ensuremath{\cProblem{U\hspace{-1pt}AST}}\xspace}
\newcommand{\PAST}{\ensuremath{\Problem{P\hspace{-2pt}AST}}\xspace}
\newcommand{\cPAST}{\ensuremath{\cProblem{P\hspace{-2pt}AST}}\xspace}
\newcommand{\UPAST}{\ensuremath{\Problem{UP\hspace{-2pt}AST}}\xspace}
\newcommand{\cUPAST}{\ensuremath{\cProblem{UP\hspace{-2pt}AST}}\xspace}
\newcommand{\EXP}{\ensuremath{\Problem{E\hspace{-1pt}XP}}}
\newcommand{\LEXP}{\ensuremath{\Problem{LE\hspace{-1pt}XP}}}
\newcommand{\REXP}{\ensuremath{\Problem{RE\hspace{-1pt}XP}}}

\newcommand{\leqm}{\mathrel{\:{\leq}_{\textnormal{m}}}}
\newcommand{\leqT}{\mathrel{\:{\leq}_{\textnormal{T}}}}
\newcommand{\equivm}{\mathrel{\:{\equiv}_{\textnormal{m}}}}
\newcommand{\equivT}{\mathrel{\:{\equiv}_{\textnormal{T}}}}

\newcommand{\LLL}[1]{\ensuremath{\textnormal{L}\big(#1\big)}\xspace}
\newcommand{\RRR}[1]{\ensuremath{\textnormal{R}\big(#1\big)}\xspace}

\mainmatter  

\title{On the Hardness of Almost--Sure Termination\thanks{This research is funded by the Excellence Initiative of the German federal and state governments and by the EU FP7 MEALS project.}}

\titlerunning{On the Hardness of Almost--Sure Termination}

%
%
\author{Benjamin Lucien Kaminski%
\and Joost-Pieter Katoen}
\authorrunning{Benjamin Lucien Kaminski \and Joost-Pieter Katoen}

\institute{Software Modeling and Verification Group\\
RWTH Aachen University\\
\mailsa}

%
%

\toctitle{On the Hardness of Almost--Sure Termination}
\tocauthor{Benjamin Lucien Kaminski, Joost-Pieter Katoen}
\maketitle

\begin{abstract}
This paper considers the computational hardness of computing expected outcomes and deciding (universal) (positive) almost--sure termination of probabilistic programs.
It is shown that computing lower and upper bounds of expected outcomes is $\Sigma_1^0$-- and $\Sigma_2^0$--complete, respectively.
Deciding (universal) almost--sure termination as well as deciding whether the expected outcome of a program equals a given rational value is shown to be $\Pi^0_2$--complete.
Finally, it is shown that deciding (universal) positive almost--sure termination is $\Sigma_2^0$--complete ($\Pi_3^0$--complete).
\end{abstract}

\allowdisplaybreaks

\section{Introduction}

Probabilistic programs~\cite{DBLP:journals/jcss/Kozen81} are imperative programs with the ability to toss a (possibly) biased coin and proceed their execution depending on the outcome of the coin toss.
They are used in randomized algorithms, in security to describe cryptographic constructions (such as randomized encryption) and security experiments~\cite{DBLP:journals/toplas/BartheKOB13}, and in machine learning to describe distribution functions that are analyzed using Bayesian inference~\cite{DBLP:journals/corr/BorgstromGGMG13}.
Probabilistic programs are typically just a small number of lines, but hard to understand and analyze, let alone algorithmically.
This paper considers the computational hardness of two main analysis problems (and variations thereof) for probabilistic programs:
\begin{enumerate}
\item Computing expected outcomes: Is the expected outcome of a program variable smaller than, equal to, or larger than a given rational number?

\item Deciding [universal] (positive) almost--sure termination: Does a program terminate [on all inputs] with probability one (within an expected finite number of computation steps)?
\end{enumerate}
The first analysis problem is related to determining weakest pre--expectations of probabilistic programs~\cite{mciver,gretz}.
Almost--sure termination is an active field of research~\cite{luis}.
A lot of work has been done towards automated reasoning for almost--sure termination.
For instance, \cite{sneyers} gives an overview of some particularly interesting examples of probabilistic logical programs and the according intuition for proving almost--sure termination. 
Arons \emph{et al.}~\cite{arons2003parameterized} reduce almost--sure termination to termination of non--deterministic programs by means of a planner.
This idea has been further exploited and refined into a pattern--based approach with prototypical tool support~\cite{esparza}. 

Despite the existence of several (sometimes automated) approaches to tackle almost--sure termination, most authors claim that it must intuitively be harder than the termination problem for ordinary programs.
To mention a few, Morgan~\cite{onlymorgan} remarks that while partial correctness for small--scale examples is not harder to prove than for ordinary programs, the case for total correctness of a probabilistic loop must be harder to analyze.
Esparza \emph{et al.}~\cite{esparza} claim that almost--sure termination must be harder to decide than ordinary termination since for the latter a topological argument suffices while for the former arithmetical reasoning is needed. 
The computational hardness of almost--sure termination has however received scant attention.
As a notable exception, \cite{tiomkin}~establishes that deciding almost--sure termination of certain concurrent probabilistic programs is in $\Pi_2^0$.


In this paper, we give precise classifications of the level of arithmetical reasoning that is needed to decide the aforementioned analysis problems by establishing the following results:
We first show that computing lower bounds on the expected outcome of a program variable $v$ after executing a probabilistic program $P$ on a given input $\eta$ is $\Sigma_1^0$--complete and therefore arbitrarily close approximations from below are computable.
Computing upper bounds, on the other hand, is shown to be $\Sigma_2^0$--complete, thus arbitrarily close approximations from above are not computable in general.
Deciding whether an expected outcome equals some rational is shown to be $\Pi_2^0$--complete.

For the second analysis problem---almost--sure termination---we obtain that deciding almost--sure termination of probabilistic program $P$ on a given input $\eta$ is $\Pi_2^0$--complete.
While for ordinary programs we have a complexity leap when moving from the non--universal to the universal halting problem, we establish that \emph{this is not the case for probabilistic programs}: Deciding universal a.s\ termination turns out to be $\Pi_2^0$--complete too.
The case for \emph{positive} almost--sure termination is different however:
While deciding (non--universal) positive almost--sure termination is $\Sigma_2^0$--complete, we show that universal positive almost--sure termination is $\Pi_3^0$--complete.


\section{Preliminaries}

As indicated, our hardness results will be stated in terms of levels in the arithmetical hierarchy---a concept we briefly recall:

\begin{definition}[Arithmetical Hierarchy \textnormal{\textbf{\cite{kleeneNF,odifreddi1}}}]
\label{remarithmetic}
For every $n \in \bbbn$, the \textbf{class $\boldsymbol{\Sigma_n^0}$} is defined as $\Sigma_n^0 = \big\{ \Problem A ~\big|~ \Problem A = \big\{  x ~\big|~ \exists y_1\, \forall y_2\, \exists y_3\, \cdots\, \exists/\forall y_n\colon~ ( x,\, y_1,\, y_2,\allowbreak\, y_3,\, \ldots,\allowbreak\, y_n) \in \Problem R\big\},\, \Problem R$ \textnormal{is a decidable relation}$\big\}$, the \textbf{class $\boldsymbol{\Pi_n^0}$} is defined as $\Pi_n^0 = \big\{ \Problem A ~\big|~ \Problem A = \big\{ x ~\big|~ \forall y_1\, \exists y_2\, \forall y_3\, \cdots\, \exists/\forall y_n\colon~ ( x,\, y_1,\, y_2,\, y_3,\, \ldots,\, y_n) \in \Problem R \big\},\, \Problem R$ \textnormal{is} \textnormal{a} \textnormal{decidable} \textnormal{re}{\-}\textnormal{la}{\-}\textnormal{tion}$\big\}$ and the \textbf{class $\boldsymbol{\Delta_n^0}$} is defined as $\Delta_n^0 = \Sigma_n^0 \cap \Pi_n^0$.
Note that we require that the values of the variables are drawn from a recursive domain.
\emph{Multiple consecutive quantifiers} \emph{of the same type} can be contracted to \emph{one} quantifier of that type, so the number $n$ really refers to the number of necessary \emph{quantifier alternations} rather than to the number of quantifiers used.
A set $\Problem A$ is called \textbf{arithmetical}, iff $\Problem A \in \Gamma_n^0$, for $\Gamma \in \{\Sigma,\, \Pi,\, \Delta\}$ and $n \in \mathbb N$.
The arithmetical sets form a strict hierarchy, i.e.\ $\Delta_n^0 \subset \Gamma_n^0 \subset \Delta_{n+1}$ and $\Sigma_n^0 \neq \Pi_n^0$ holds for $\Gamma \in \{\Sigma,\, \Pi\}$ and $n \geq 1$.
Furthermore, note that $\Sigma_0^0 = \Pi_0^0 = \Delta_0^0 = \Delta_1^0$ is exactly the class of the decidable sets and $\Sigma_1^0$ is exactly the class of the recursively enumerable sets.
\end{definition}
Next, we recall the concept of many--one reducibility and completeness:
\begin{definition}[Many--One Reducibility and Completeness \textnormal{\textbf{\cite{odifreddi1,post44,davis}}}]
Let $\Problem A,\, \Problem B$ be arithmetical sets and let $X$ be some appropriate universe such that $\Problem A,\Problem B \subseteq X$.
$\Problem A$ is called \textbf{many--one--reducible} to $\Problem B$, denoted $\boldsymbol{\Problem A \leqm \Problem B}$, iff there exists a computable function $f\colon X \rightarrow X$, such that $\forall\, {x} \in X\colon \big( x \in \Problem A \Longleftrightarrow f( x) \in \Problem B\big)$.
If $f$ is a function such that $f$ many--one reduces $\Problem A$ to $\Problem B$, we denote this by $\boldsymbol{f\colon \Problem A \leqm \Problem B}$.
Note that $\leqm$ is transitive. 

$\Problem A$ is called \textbf{$\boldsymbol{\Gamma_n^0}$--complete}, for $\Gamma \in \{\Sigma,\, \Pi,\, \Delta\}$, iff both $\Problem A \in \Gamma_n^0$ and $\Problem A$ is \textbf{$\boldsymbol{\Gamma_n^0}$--hard}, meaning $\Problem C \leqm \Problem A$, for any set $\Problem C \in \Gamma_n^0$. Note that if $\Problem A$ is $\Gamma_n^0$--complete and $\Problem A \leqm \Problem B$, then $\Problem B$ is necessarily $\Gamma_n^0$--hard.
Furthermore, note that if $\Problem A$ is $\Sigma_n^0$--complete, then $\Problem{A} \in \Sigma_n^0\setminus \Pi_n^0$. Analogously if $\Problem A$ is $\Pi_n^0$--complete, then $\Problem{A} \in \Pi_n^0\setminus \Sigma_n^0$.
\end{definition}
%
%

\section{Probabilistic Programs}

In order to speak about probabilistic programs and the computations performed by such programs, we briefly introduce the syntax and semantics we use:
\begin{definition}[Syntax]
Let \sovars be the set of program variables.
The \textbf{set $\boldsymbol{\textsf{\textbf{Prog}}}$ of probabilistic programs} adheres to the following grammar:
\begin{align*}
\soprogs ~~\longrightarrow~~ &v\mathrel{\ttt{:=}}e ~|~ \soprogs\ttt{;}\: \soprogs ~|~ \ttt{\{}\soprogs\ttt{\}} \:[p]\: \ttt{\{}\soprogs\ttt{\}} ~|~ \ttt{WHILE}\:\ttt{($b$)}\:\ttt{\{} \soprogs \ttt{\}}~,
\end{align*}
where $v \in \sovars$, $e$ is an arithmetical expression over \sovars, $p \in [0,\, 1] \subseteq \bbbq$, and $b$ is a Boolean expression over arithmetic expressions over $\sovars$. 
We call the set of programs that \emph{do not} contain any probabilistic choices the \textbf{set of ordinary programs} and denote this set by \textsf{\textbf{ordProg}}.
\end{definition}
The presented syntax is the one of the fully probabilistic\footnote{Fully probabilistic programs may contain probabilistic but no non--deterministic choices.} fragment of the probabilistic guarded command language (pGCL) originally due to McIver  and Morgan~\cite{mciver}.
We omitted \texttt{skip}--, \texttt{abort}--, and \texttt{if}--statements, as those are syntactic sugar. 
While assignment, concatenation, and the while--loop are standard programming constructs, $\ttt{\{}P_1\ttt{\}} \:[p]\: \ttt{\{}P_2\ttt{\}}$ denotes a probabilistic choice between programs $P_1$ (with probability $p$) and $P_2$ (with probability $1-p$).
An operational semantics for pGCL programs is given below:
\begin{definition}[Semantics]
\label{def:semantics}
Let the set of variable valuations be denoted by $\mathbb V = \{\eta ~|~ \eta\colon \sovars \rightarrow \bbbq^+\}$, let the set of program states be denoted by $\bbbs = \big(\soprogs \cup \{{\downarrow}\}\big) \times \mathbb V \times I \times \{L,\, R\}^*$, for $I = [0,\, 1] \cap \bbbq^+$, let $\llbracket e \rrbracket_\eta$ be the evaluation of the arithmetical expression $e$ in the variable valuation $\eta$, and analogously let $\llbracket b \rrbracket_\eta$ be the evaluation of the Boolean expression $b$. 
Then the \textbf{semantics of probabilistic programs} is given by the smallest relation ${\vdash} \subseteq \bbbs \times \bbbs$ which satisfies the following inference rules:
\normalsize\begin{align*}
 (\textnormal{assign})&\frac{}{\langle v\textnormal{\texttt{ := }} e,\, \eta,\, a,\, \theta\rangle ~\vdash~  \langle {\downarrow},\, \eta[v \mapsto \max \{\llbracket e \rrbracket_\eta,\, 0\}],\, a,\, \theta\rangle}\\[0.5\baselineskip]
(\textnormal{concat1})&\frac{\langle P_1,\, \eta,\, a,\, \theta\rangle ~\vdash~  \langle P_1',\, \eta',\, a',\, \theta'\rangle}{\langle P_1\textnormal{\texttt{;}}\:P_2,\, \eta,\, a,\, \theta\rangle ~\vdash~  \langle P_1'\textnormal{\texttt{;}}\: P_2,\, \eta',\, a',\, \theta'\rangle}\\[0.5\baselineskip]
(\textnormal{concat2})&\frac{}{\langle {\downarrow}\textnormal{\texttt{;}}\:P_2,\, \eta,\, a,\, \theta\rangle ~\vdash~  \langle P_2,\, \eta,\, a,\, \theta\rangle}\\[0.5\baselineskip]
(\textnormal{prob1})&\frac{}{\langle \{P_1\}\:[p]\:\{P_2\},\, \eta,\, a,\, \theta\rangle ~\vdash~  \langle P_1,\, \eta,\, a \cdot p,\, \theta\cdot L \rangle}\\[0.5\baselineskip]
(\textnormal{prob2})&\frac{}{\langle \{P_1\}\:[p]\:\{P_2\},\, \eta,\, a,\, \theta\rangle ~\vdash~  \langle P_2,\, \eta,\, a \cdot (1 - p),\, \theta \cdot R\rangle}\\[0.5\baselineskip]
(\textnormal{while1})&\frac{\llbracket b \rrbracket_\eta = \textnormal{True}}{\langle \textnormal{\texttt{WHILE}}\:\ttt{(} b \textnormal{\texttt{)}}\:\texttt{\{} P \textnormal{\texttt{\}}},\, \eta,\, a,\, \theta\rangle ~\vdash~  \langle P\textnormal{\texttt{;}}\: \ttt{WHILE}\:\ttt{(} b \textnormal{\texttt{)}}\:\ttt{\{} P \textnormal{\texttt{\}}},\, \eta,\, a,\, \theta\rangle}\\[0.5\baselineskip]
(\textnormal{while2})&\frac{\llbracket b \rrbracket_\eta = \textnormal{False}}{\langle \textnormal{\texttt{WHILE}}\:\ttt{(} b \textnormal{\texttt{)}}\:\ttt{\{} P \textnormal{\texttt{\}}},\, \eta,\, a,\, \theta\rangle ~\vdash~  \langle {\downarrow},\, \eta,\, a,\, \theta\rangle}
\end{align*}\normalsize
We use $\sigma  \vdash^k \tau$ in the usual sense.
\end{definition}
The semantics is mostly straightforward except for two features: in addition to the program that is to be executed next and the current variable valuation, each state also stores a sequence $\theta$ that encodes which probabilistic choices were made in the past ($L$eft or $R$ight) as well as the probability $a$ that those choices were made.
The graph that is spanned by the $\vdash$--relation can be seen as an unfolding of the Markov decision process semantics for pGCL provided by Gretz \emph{et al.}~\cite{gretz} when restricting oneself to fully probabilistic programs.

\section{Expected Outcomes and Termination Probabilities}

In this section we formally define the notion of an expected outcome as well the notion of (universal) (positive) almost--sure termination.
We start by investigating how state successors can be computed.

It is a well--known result due to Kleene that for any ordinary program $P$ and a state $\sigma$ the $k$-th successor of $\sigma$ with respect to $\vdash$ is unique and computable.
If, however, $P$ is a probabilistic program containing probabilistic choices, the $k$-th successor of a state need not be unique, because at various points of the execution the program must choose a left or a right branch with some probability. 
However, if we resolve those choices by providing a sequence of symbols $w$ over the alphabet $\{L,\, R\}$ that encodes for all probabilistic choices which occur whether the $L$eft or the $R$ight branch shall be chosen at a branching point, we can construct a computable function that computes a unique $k$-th successor. 
Notice that for this purpose a sequence of finite length is sufficient.
We obtain the following:
\begin{proposition}[The State Successor Function]
\label{statesucc}
Let $\bbbs_\bot = \bbbs \cup \{\bot\}$.
There exists a total computable function $\T\colon \bbbn \times \bbbs \times \{L,\, R\}^* \rightarrow \bbbs_\bot$, such that for $k \geq 1$
\normalsize
\belowdisplayskip=0pt\begin{align*}
\T_0(\sigma,\, w) ~&=~ \begin{cases}
\sigma, &\textnormal{if }w = \varepsilon,\\
\bot, &\textnormal{otherwise,}
\end{cases}\\
\T_{k}(\sigma,\, w) ~&=~ \begin{cases}
T_{k-1}(\tau,\, w'), &\textnormal{if }\sigma = \langle P,\, \eta,\, a,\, \theta\big\rangle \vdash \langle P',\, \eta',\, a',\, \theta \cdot b\big\rangle = \tau,\\
& ~~~  \textnormal{with } w = b\cdot w' \textnormal{ and } b \in \{L,\, R,\, \varepsilon\},\\
\bot &\textnormal{otherwise.}
\end{cases}
\end{align*}\normalsize
\end{proposition}
So $\T_k(\sigma,\, w)$ returns a successor state $\tau$, if $\sigma \vdash^k \tau$, whereupon exactly $|w|$ inferences must use the \textnormal{(prob1)}-- or the \textnormal{(prob2)}--rule and those probabilistic choices are resolved according to $w$. 
Otherwise $\T_k(\sigma,\, w)$ returns $\bot$.
Note in particular that for both the inference of a terminal state $\langle{\downarrow},\, \eta,\, a,\, \theta\rangle$ within less than $k$ steps as well as the inference of a terminal state through less or more than $|w|$ probabilistic choices, the calculation of $\T_k(\sigma,\, w)$ will result in $\bot$.
In addition to $\T$, we will need two more computable operations for expressing expected outcomes, termination probabilities, and expected runtimes:
\begin{proposition}
\label{prop:statesucc}
There exist two total computable functions $\alpha\colon \bbbs_\bot \rightarrow \bbbq^+$ and $\wp\colon \bbbs_\bot \times \sovars  \rightarrow \bbbq^+$, such that
\normalsize\begin{align*}
\alpha(\sigma) ~~&{=}~\begin{cases}
a,&\textnormal{if } \sigma = \langle {\downarrow},\, \underline{~\:},\,  a,\, \underline{~\:}\rangle\\
0,&\textnormal{otherwise,}
\end{cases}
\quad~\:~\wp(\sigma,\,  v) ~~{=}~\begin{cases}
\eta(v) \cdot a,&\textnormal{if } \sigma = \langle {\downarrow},\, \eta,\,  a,\, \underline{~\:}\rangle\\
0,&\textnormal{otherwise,}
\end{cases}
\end{align*}\normalsize
where $\underline{~\:}$ represents an arbitrary value.
\end{proposition}
The function $\alpha$ takes a state $\sigma$ and returns the probability of reaching $\sigma$.
The function $\wp$ takes a state $\sigma$ and a variable $v$ and returns the probability of reaching $\sigma$ multiplied with the value of $v$ in the state $\sigma$.
Both functions do that only if the provided state $\sigma$ is a terminal state. 
Otherwise they return 0.
Based on the above notions, we now definie expected outcomes, termination probabilities and expected times until termination:
\begin{definition}[Expected Outcome, Termination Probability, and Expected Time until Termination]
\label{exp}
Let $P \in \soprogs$, $\eta \in \mathbb V$, $v \in \sovars$, $\sigma_{P,\eta} = \langle P,\, \eta,\, 1,\, \varepsilon\rangle$, and for a finite alphabet $A$ let $A^{{}\leq k} = \bigcup_{i=0}^k A^i$.
Then
\begin{enumerate}
\item the \textbf{expected outcome} of $v$ after executing $P$ on $\eta$, denoted $\boldsymbol{\textnormal{\textbf{E}}_{P,\eta}(v)}$, is
\abovedisplayskip=5pt
\begin{align*}
\Exp{P,\eta}{v} ~=~ \sum_{k = 0}^{\infty} ~ \sum_{w \in \{L,\, R\}^{{}\leq k}} ~ \wp\big(\T_{k}(\sigma_{P,\eta},\, w),\, v \big)~,
\end{align*}\normalsize
\item the \textbf{probability that $\boldsymbol P$ terminates} on $\eta$, denoted $\boldsymbol{\textnormal{\textbf{Pr}}_{P,\eta}({\downarrow})}$, is
\abovedisplayskip=5pt
\begin{align*}
\Pr_{P,\eta}({\downarrow}) ~=~ \sum_{k = 0}^{\infty} ~ \sum_{w \in \{L,\, R\}^{{}\leq k}} ~ \alpha\big(\T_{k}(\sigma_{P,\eta},\, w) \big)~,
\end{align*}
\item the \textbf{expected time until termination of $\boldsymbol P$} on $\eta$, denoted $\boldsymbol{\textnormal{\textbf{E}}_{P,\eta}({\downarrow})}$, is
\abovedisplayskip=5pt
\begin{align*}
\Exp{P,\eta}{\downarrow} ~=~ \sum_{k = 0}^{\infty} \left(1-\sum_{w \in \{L,\, R\}^{{}\leq k}} ~ \alpha\big(\T_{k}(\sigma_{P,\eta},\, w) \big)\right)~.
\end{align*}\normalsize
\end{enumerate}
\end{definition}
The expected outcome $\Exp{P,\eta}{v}$ as defined here coincides with the weakest pre--expectation $\mathit{wp}.P.v\,(\eta)$ \`{a} la McIver and Morgan \cite{mciver} for fully probabilistic programs.
In the above definition for $\Exp{P,\eta}{v}$, we sum over all possible numbers of inference steps $k$ and sum over all possible sequences from length $0$ up to length $k$ for resolving all probabilistic choices. 
Using $\wp$ we filter out the terminal states $\sigma$ and sum up the values of $\wp(\sigma,\, v)$.

For the termination probability $\Pr_P({\downarrow})$, we basically do the same but we merely sum up the probabilities of reaching final states by using $\alpha$ instead of $\wp$. 

For the expected time until termination $\Exp{P,\eta}{\downarrow}$, we go along the lines of~\cite{luis}: It is stated there that the expected time until termination of $P$ on $\eta$ can be expressed as $\sum_{k=0}^{\infty} \Pr(\textnormal{``}P$ runs for more than $k$ steps on $\eta\textnormal{"}) = \sum_{k=0}^{\infty}\big( 1 - \Pr(\textnormal{``}P$ terminates within $k$ steps on $\eta\textnormal{"})\big)$.
We have expressed the latter in our set--up.

In order to investigate the complexity of calculating $\Exp{P,\eta}{v}$, we define three sets: 
$\LEXP$, which relates to the set of rational lower bounds of $\Exp{P,\eta}{v}$, $\REXP$, which relates to the set of rational upper bounds, and $\EXP$ which relates to the value of $\Exp{P,\eta}{v}$ itself:
\begin{definition}[$\boldsymbol\LEXP$, $\boldsymbol\REXP$, and $\boldsymbol\EXP$]
The sets $\LEXP,\REXP,\EXP \subset \soprogs \times \mathbb V \times \sovars \times \bbbq^+$ are defined as $(P,\, \eta,\, v,\, q) \in \LEXP$ iff $q < \Exp{P,\eta}{v}$, $(P,\, \eta,\, v,\, q) \in \REXP$ iff $q > \Exp{P,\eta}{v}$, and
$(P,\, \eta,\, v,\, q) \in \EXP$ iff $q = \Exp{P,\eta}{v}$.
\end{definition}
Regarding the termination probability of a probabilistic program, the case of almost--sure termination is of special interest:
We say that a program $P$ \emph{terminates almost--surely} on input $\eta$ iff $P$ terminates on $\eta$ with probability 1.
Furthermore, we say that $P$ \emph{terminates positively almost--surely} on $\eta$ iff the expected time until termination of $P$ on $\eta$ is finite.
Lastly, we say that $P$ terminates \emph{universally} (positively) almost--surely, if it does so on all possible inputs $\eta$.
The problem of (universal) almost--sure termination can be seen as the probabilistic counterpart to the (universal) halting problem for ordinary programs.

In the following, we formally define the according problem sets:
\begin{definition}[Almost--Sure Termination Problem Sets]
The sets $\boldsymbol\AST$, $\boldsymbol\PAST$, $\boldsymbol\UAST$, and $\boldsymbol\PAST$ are defined as follows:
\small\begin{align*}
(P,\, \eta) \in \AST ~\Longleftrightarrow~ \Pr_{P,\eta}({\downarrow}) = 1~~~~~~~\qquad
&(P,\, \eta) \in \PAST ~\Longleftrightarrow~ \Exp{P,\eta}{\downarrow} < \infty\\
P \in \UAST ~\Longleftrightarrow~ \forall \eta\colon (P,\, \eta) \in \AST\qquad
&~~~P \in \UPAST ~\Longleftrightarrow~ \forall \eta\colon (P,\, \eta) \in \PAST
\end{align*}\normalsize
Notice that both $\PAST \subset \AST$ and $\UPAST \subset \UAST$ hold.
\end{definition}

\section{The Hardness of Computing Expected Outcomes}
In this section we investigate the computational hardness of deciding the sets $\LEXP$, $\REXP$, and $\EXP$.
The first fact we establish is the $\Sigma_1^0$--completeness of $\LEXP$. This result is established by reduction from the (non--universal) halting problem for ordinary programs:
\begin{theorem}[The Halting Problem \textnormal{\textbf{\cite{odifreddi2}}}]
\label{HPcomp}
The \textbf{halting problem} is a subset $\boldsymbol\HP \subset \sonprogs \times \mathbb{V}$, which is characterized as $(P,\, \eta) \in \HP$ iff \:$\exists k \: \exists \eta'\colon \T_k(\sigma_{P,\eta},\allowbreak\, \varepsilon) = \langle {\downarrow},\, \eta',\, 1,\, \varepsilon\rangle$.
Let $\boldsymbol\cHP$ denote the \textbf{complement of the halting problem}, i.e.\ $\cHP = (\sonprogs \times \mathbb V) \setminus \HP$.
$\HP$ is $\Sigma_1^0$--complete and $\cHP$ is $\Pi_1^0$--complete.
\end{theorem}
\begin{theorem}
\label{LisSigma1comp}
$\LEXP$ is $\Sigma_1^0$--complete.
\end{theorem}
\begin{proof}
For $\LEXP \in \Sigma_1^0$ consider the following:
\begin{align*}
~&(P,\, \eta,\, v,\,  q) \in \LEXP\\
\Longleftrightarrow~& q < \Exp{P,\eta}{v}\\
\Longleftrightarrow~&q < \sum_{k = 0}^{\infty} ~ \sum_{w \in \{L,\, R\}^{{}\leq k}} ~ \wp\big(\T_{k}(\sigma_{P,\eta},\, w),\, v \big)\\
\Longleftrightarrow~&\exists \, y\colon q < \sum_{k = 0}^{y} ~ \sum_{w \in \{L,\, R\}^{{}\leq k}} ~ \wp\big(\T_{k}(\sigma_{P,\eta},\, w),\, v \big)\\
\Longrightarrow~&\LEXP \in \Sigma_1^0
\end{align*}
Figure \ref{fig} (left) gives an intuition on this formula.

It remains to show that $\LEXP$ is $\Sigma_1^0$--hard:
We do this by proving $\HP \leqm \LEXP$.
Consider the following function $f\colon \HP \leqm \LEXP$: $f$ takes an ordinary program $Q \in \sonprogs$ and a variable valuation $\eta$ as its input and computes $(P,\, \eta,\, v,\, \nicefrac 1 2)$, where $P$ is the probabilistic program $v\mathrel{\texttt{:=}}0\texttt{;}\:\{v\mathrel{\texttt{:=}}1\}[\nicefrac 1 2]\{\mathit{TQ}\texttt{;}\:v\mathrel{\texttt{:=}}1\}$ and $\mathit{TQ}$ is an ordinary program that simulates $Q$ on $\eta$.

\textit{Correctness of the reduction:} There are two cases:
(1) $Q$ terminates on input $\eta$. 
Then the expected outcome of variable $v$ after executing the program $P$ on input $\eta$ is 1, because in both branches, the variable will be set to 1.
As $\nicefrac 1 2 < 1$, we have that $(P,\, \eta,\, v,\, \nicefrac 1 2) \in \LEXP$.

(2) $Q$ does not terminate on input $\eta$. 
Then the expected outcome of variable $v$ after executing the program $P$ on input $\eta$ is $\nicefrac 1 2$, because only in the left branch (which has probability $\nicefrac 1 2$), the variable will be set to 1.
In the right branch, the program does not terminate and therefore the outcome of this branch is 0.
As $\nicefrac 1 2 \not< \nicefrac 1 2$, we have that $(P,\, \eta,\, v,\, \nicefrac 1 2) \not\in \LEXP$.
%
%
\qed
\end{proof}
Theorem \ref{LisSigma1comp} implies that $\LEXP$ is recursively enumerable.
This means that all lower bounds for expected outcomes can be effectively enumerated by some algorithm.
Now, if upper bounds were recursively enumerable as well, then expected outcomes would be computable reals.
However, the contrary will be shown by establishing that $\REXP$ is $\Sigma_2^0$--complete, thus $\REXP \not\in \Sigma_1^0$ and hence $\REXP$ is not recursively enumerable.
$\Sigma_2^0$--hardness will be established by a reduction from the  complement of the universal halting problem for ordinary programs:
\begin{theorem}[The Universal Halting Problem \textnormal{\textbf{\cite{odifreddi2}}}]
\label{UHPcomplete}
The \textbf{universal halting problem} is a subset $\boldsymbol\UHP \subset \sonprogs$, which is characterized as $P \in \UHP$ iff $\forall \,\eta\colon (P,\, \eta) \in \HP$.
Let $\boldsymbol\cUHP$ denote the \textbf{complement of $\boldsymbol\UHP$}, i.e., $\cUHP = \sonprogs \setminus \UHP$.
$\UHP$ is $\Pi_2^0$--complete and $\cUHP$ is $\Sigma_2^0$--complete.
\end{theorem}
\begin{theorem}
\label{RisSigma2comp}
$\REXP$ is $\Sigma_2^0$--complete.
\end{theorem}
\begin{proof}
For $\REXP \in \Sigma_2^0$ consider the following:
\begin{align*}
~&(P,\, \eta,\, v,\,  q) \in \REXP\\
\Longleftrightarrow~&q > \Exp{P,\eta}{v}\\
\Longleftrightarrow~&q > \sum_{k = 0}^{\infty} ~ \sum_{w \in \{L,\, R\}^{{}\leq k}} ~ \wp\big(\T_{k}(\sigma_{P,\eta},\, w),\, v \big)\\
\Longleftrightarrow~&\exists\, \delta \: \forall\, y\colon q - \delta > \sum_{k = 0}^{y} ~ \sum_{w \in \{L,\, R\}^{{}\leq k}} ~ \wp\big(\T_{k}(\sigma_{P,\eta},\, w),\, v \big)\\
\Longrightarrow~&\REXP \in \Sigma_2^0
\end{align*}\normalsize
Figure \ref{fig} (right) gives an intuition on this formula.

It remains to show that $\REXP$ is $\Sigma_2^0$--hard:
We do this by proving $\cUHP \leqm \REXP$.
Consider the following function $f\colon \cUHP \leqm \REXP$: $f$ takes an ordinary program $Q \in \sonprogs$ as its input and computes the triple $(P,\, \eta,\, v,\, 1)$, where $\eta$ is an arbitrary but fixed input, and $P \in \soprogs$ is the following probabilistic program:
\begin{alltt}
\(i\) := 0;  \{\(c\) := 0\} [0.5] \{\(c\) := 1\};
while (\(c\) \(\neq\) 0)\{ \(i\) := \(i\) + 1; \{\(c\) := 0\} [0.5] \{\(c\) := 1\} \};

\(k\) := 0; \{\(c\) := 0\} [0.5] \{\(c\) := 1\};
while (\(c\) \(\neq\) 0)\{ \(k\) := \(k\) + 1; \{\(c\) := 0\} [0.5] \{\(c\) := 1\} \};

\(v\) := 0; \(TQ\)
\end{alltt}
$TQ$ is a program that computes $\alpha \big(\T_k\big(\big\langle Q,\, g_Q(i),\, 1,\, \varepsilon\big\rangle,\, \varepsilon\big)\big) \cdot 2^{k+1}$ and stores the result in the variable $v$, and $g_Q\colon \bbbn \rightarrow \mathbb V$ is some computable bijection, such that $\forall z \in \sovars\colon \big(g_Q(i)\big)(z) \neq 0$ implies that $z$ occurs in $Q$.

\textit{Correctness of the reduction:}  $\alpha \big(\T_k\big(\big\langle Q,\, g_Q(i),\, 1,\, \varepsilon\big\rangle,\, \varepsilon\big)\big) \cdot 2^{k+1}$ returns $2^{k+1}$ if and only if $Q$ halts on input $g_Q(i)$ after exactly $k$ steps (otherwise it returns 0).
 The two while--loops generate independent geometric distributions with parameter $\nicefrac 1 2$ on $i$ and $k$, respectively, so the probability of generating exactly the numbers $i$ and $k$ is $\nicefrac{1}{2^{k+1}} \cdot \nicefrac{1}{2^{k+1}} = \nicefrac{1}{2^{i+k+2}}$. 
The expected outcome of $v$ after executing the program $P$ on any input $\eta$ is hence
\begin{align*}
&\sum_{i= 0}^{\infty} \sum_{k= 0}^{\infty} \frac{1}{2^{i+k+2}} \cdot \alpha\bigg(\T_k\Big(\big\langle Q,\, g_Q(i),\, 1,\, \varepsilon \big\rangle,\, \varepsilon\Big)\bigg) \cdot 2^{k+1}~.
\end{align*}
Since for each input, the number of steps until termination of $Q$ is either unique or does not exist, the formula for the expected outcome reduces to $\sum_{i = 0}^{\infty} \frac{1}{2^{i+1}} = 1$ if and only if $Q$ halts on every input after some finite number of steps.
Thus if there exists an input on which $Q$  \emph{does not} eventually halt, then $(P,\, \eta,\, v,\, 1) \in \REXP$ as then the expected value is strictly less than one.
If, on the other hand, $Q$ \emph{does} halt on every input, then the expected outcome is exactly one and hence $(P,\, \eta,\, v,\, 1) \not\in \REXP$.
%
%
\qed
\end{proof}
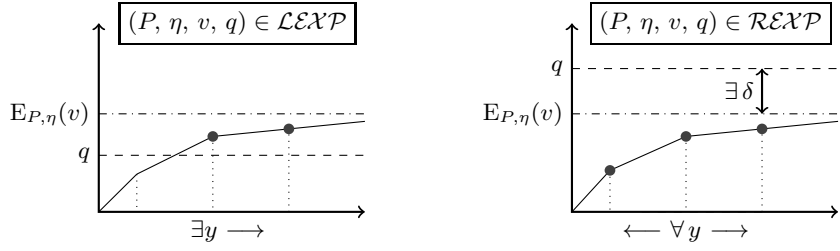
\begin{figure}[t!]
\begin{center}
\begin{tikzpicture}[scale=1]
    \draw[white, use as bounding box] (0,0.1) rectangle (3.5,2.8);
    \draw [<->,thick] (0,2.5) node (yaxis) [above] {}
        |- (3.5,0) node (xaxis) [right] {};
    \draw (0,0) coordinate (a_1) -- (0.5,.5) coordinate (a_2);
    \draw (.5,.5) coordinate (a_3) -- (1.5,1) coordinate (a_4);
    \draw (1.5,1) coordinate (a_5) -- (3.5,1.2) coordinate (a_6);
    \draw[dashed] (0,.75) coordinate (b_1) -- (3.5,.75) coordinate (b_2) node[left] at (0,.75) {$q$};
    \coordinate (c1) at (0.5, 0);
    \coordinate (c2) at (0.5, 1);
    \coordinate (c) at (intersection of a_1--a_2 and c1--c2);
    \draw[dotted] (c) -- (xaxis -| c);
    \coordinate (cc1) at (1.5, 0);
    \coordinate (cc2) at (1.5, 1);
    \coordinate (cc) at (intersection of a_3--a_4 and cc1--cc2);
    \draw[dotted] (cc) -- (xaxis -| cc) node[below] at (1.70, 0){${\exists y}~ {\longrightarrow}$};
    \coordinate (ccc1) at (2.5, 0);
    \coordinate (ccc2) at (2.5, 1);
    \coordinate (ccc) at (intersection of a_5--a_6 and ccc1--ccc2);
    \draw[dotted] (xaxis -| ccc) -- (ccc);
    
    \draw[dashdotted] (0, 1.3) coordinate (eeeee1) -- (3.5, 1.3) coordinate (eeeee2);
    \node[left] at (0, 1.3) {$\Exp{P,\eta}{v}$};
    \fill[darkgray] (cc) circle (2pt);
    \fill[darkgray] (ccc) circle (2pt);
    \node[] at (1.825, 2.5) {\fbox{$(P,\, \eta,\, v,\, q) \in \LEXP$}};
\end{tikzpicture}
\qquad\qquad\qquad\qquad
\begin{tikzpicture}[scale=1]
    \draw[white, use as bounding box] (0,0.1) rectangle (3.5,2.8);
    \draw [<->,thick] (0,2.5) node (yaxis) [above] {}
        |- (3.5,0) node (xaxis) [right] {};
    \draw (0,0) coordinate (a_1) -- (0.5,.55) coordinate (a_2);
    \draw (.5,.55) coordinate (a_3) -- (1.5,1) coordinate (a_4);
    \draw (1.5,1) coordinate (a_5) -- (3.5,1.2) coordinate (a_6);
    \draw[dashed] (0, 1.9) coordinate (b_1) -- (3.5, 1.9) coordinate (b_2) node[left] at (0, 1.9) {$q$};
     \draw[dashdotted] (0, 1.3)  coordinate (e1) -- (3.5, 1.3) coordinate (e2);
    
    \coordinate (c1) at (0.5, 0);
    \coordinate (c2) at (0.5, 1);
    \coordinate (c) at (intersection of a_1--a_2 and c1--c2);
    \draw[dotted] (c) -- (xaxis -| c);
    \coordinate (cc1) at (1.5, 0);
    \coordinate (cc2) at (1.5, 1);
    \coordinate (cc) at (intersection of a_3--a_4 and cc1--cc2);
    \draw[dotted] (cc) -- (xaxis -| cc) node[below] {${\longleftarrow} ~ {\forall\, y} ~ {\longrightarrow}$};
    \coordinate (ccc1) at (2.5, 0);
    \coordinate (ccc2) at (2.5, 1);
    \coordinate (ccc) at (intersection of a_5--a_6 and ccc1--ccc2);
    \draw[dotted] (xaxis -| ccc) coordinate (dac) -- (ccc) coordinate (dacc);
    
    \draw[<->, thick] (2.5,1.3) -- (2.5, 1.9) node at (2.2, 1.6) {$\exists\, \delta$};

     \draw[dashed] (ccc1) -- (.875,0);
    \fill[darkgray] (ccc) circle (2pt);
     \fill[darkgray] (cc) circle (2pt);
      \fill[darkgray] (c) circle (2pt);

   
    \node[left] at (0, 1.3) {$\Exp{P,\eta}{v}$};

     \node[] at (1.825, 2.5) {\fbox{$(P,\, \eta,\,v,\, q) \in \REXP$}};
\end{tikzpicture}
\end{center}
\caption{Schematic depiction of the formulae defining $\LEXP$ and $\REXP$, respectively. In each diagram, the solid line represents the monotonically increasing graph of $\sum_{k = 0}^{y} \: \sum_{w \in \{L,\, R\}^{{}\leq k}} \: \wp\big(\T_{k}(\sigma_{P,\eta},\, w),\, v \big)$ plotted over increasing $y$.}
\label{fig}
\end{figure}
Finally, we establish the following result regarding exact expected outcomes:
\begin{theorem}
\label{EisPi2comp}
$\EXP$ is $\Pi_2^0$--complete.
\end{theorem}
\begin{proof}
For $\EXP \in \Pi_2^0$ consider the following: By Theorem \ref{RisSigma2comp}, there exists a decidable relation $\Problem R$, such that $(P,\, \eta,\, v,\, x) \in \REXP$ iff $\exists r_1 \forall r_2 \colon (r_1,\, r_2,\, P,\, \eta,\, v,\, x) \in \Problem R$.
Furthermore from Theorem \ref{LisSigma1comp} it follows that there exists a decidable relation $\Problem L$, such that $(P,\, \eta,\, v,\, x) \in \LEXP$ iff $\exists \ell \colon (\ell,\, P,\, \eta,\, v,\, x) \in \Problem L$.
Let $\neg \Problem R$ and $\neg \Problem L$ be the (decidable) negations of $\Problem R$ and $\Problem L$, respectively, then:
\begin{align*}
~&(P,\, \eta,\, v,\, q) \in \EXP\\
\Longleftrightarrow~&q = \Exp{P,\eta}{v}\\
\Longleftrightarrow~&q \leq \Exp{P,\eta}{v} \,\wedge\, q \geq \Exp{P,\eta}{v}\\
\Longleftrightarrow~&\neg \big(q > \Exp{P,\eta}{v}\big) \,\wedge\, \neg \big(q < \Exp{P,\eta}{v}\big)\\
\Longleftrightarrow~&\neg \big(\exists r_1\: \forall r_2 \colon (r_1,\, r_2,\, P,\, \eta,\, v,\, q) \in \Problem R\big) \,\wedge\, \neg \big(\exists \ell \colon (\ell,\, P,\, \eta,\, v,\, q) \in \Problem L\big)\\
\Longleftrightarrow~&\big(\forall r_1\: \exists r_2 \colon (r_1,\, r_2,\, P,\, \eta,\, v,\, q) \in \neg \Problem R\big) \,\wedge\, \big(\forall \ell \colon (\ell,\, P,\, \eta,\, v,\, q) \in \neg \Problem L\big)\\
\Longleftrightarrow~&\forall r_1\: \forall \ell\: \exists r_2 \colon (r_1,\, r_2,\, P,\, \eta,\, v,\, q) \in \neg \Problem R \,\wedge\, (\ell,\, P,\, \eta,\, v,\, q) \in \neg \Problem L\\
\Longrightarrow~&\EXP \in \Pi_2^0
\end{align*}
It remains to show that $\EXP$ is $\Pi_2^0$--hard.
We do this by proving $\UHP \leqm \EXP$.
Consider again the reduction function $f$ from the proof of Theorem \ref{RisSigma2comp}: Given an ordinary program $Q$, $f$ computes the triple $(P,\, \eta,\, v,\, 1)$, where $P$ is a probabilistic program which has an expected outcome of one for the variable $v$ if and only if $Q$ terminates on all inputs, which is nothing else than $Q \in \UHP$.
Thus $f\colon \UHP \leqm \EXP$.
%
\qed
\end{proof}
%
%

\section{The Hardness of Deciding Probabilistic Termination}
This section presents the main contributions of this paper: Hardness results on several variations of almost--sure termination problems.
We first establish that deciding almost--sure termination of a program on a given input is $\Pi_2^0$--complete:
\begin{theorem}
\label{ASTisPi2comp}
$\AST$ is $\Pi_2^0$--complete.
\end{theorem}
\begin{proof}
For $\AST \in \Pi_2^0$ we show $\AST \leqm \EXP$.
For that, consider the following function $f\colon \AST \leqm \EXP$ which takes a probabilistic program $Q$ and an input $\eta$ as its input and computes the tuple $(P,\, \eta,\, v,\, 1)$, where $v \in \sovars$ does not occur in $Q$ and $P$ is the probabilistic program $v \mathrel{\ttt{:=}} 0\ttt{;}\: Q\ttt{;}\: v \mathrel{\ttt{:=}} 1$

\textit{Correctness of the reduction:}
On executing $P$, the variable $v$ is set to one only in those runs in which the program $Q$ terminates on $\eta$. So the expected value of $v$ converges to one, if and only if the probability of $Q$ terminating converges to one.
So if $(Q,\, \eta) \in \AST$, then and only then $(P,\, \eta,\, v,\, 1) \in \EXP$.
Thus $f\colon \AST \leqm \EXP$~.
By Theorem \ref{EisPi2comp}, $\EXP$ is $\Pi_2^0$--complete, so it follows directly that $\AST \in \Pi_2^0$.

It remains to show that $\AST$ is $\Pi_2^0$--hard.
For that we many--one reduce the $\Pi_2^0$--complete universal halting problem to $\AST$ using the following function $f'\colon \UHP \leqm \AST$: $f'$ takes an ordinary program $Q$ as its input and computes the pair $(P',\, \eta)$, where $\eta$ is some arbitrary but fixed input and $P'$ is the following probabilistic program:
\begin{alltt}
\(i\) := 0; \{\(c\) := 0\} [0.5] \{\(c\) := 1\};
while (\(c\) \(\neq\) 0)\{ \(i\) := \(i\) + 1; \{\(c\) := 0\} [0.5] \{\(c\) := 1\} \};
\(SQ\)
\end{alltt}
where $SQ$ is an ordinary program that on any input $\eta$ simulates the program $Q$ on input $g_Q(i)$, and $g_Q\colon \bbbn \rightarrow \mathbb V$ is some computable bijection, such that $\forall v \in \sovars\colon \big(g_Q(i)\big)(v) \neq 0$ implies that $v$ occurs in $Q$.

\textit{{Correctness of the reduction:}}
The while--loop in $P'$ establishes a geometric distribution with parameter $\nicefrac 1 2$ on $i$ and hence a geometric distribution on all possible inputs for $Q$.
After the while--loop, the program $Q$ is simulated on the input generated probabilistically in the while--loop.
Obviously then the entire program $P'$ terminates with probability one on any arbitrary input $\eta$, i.e.\ terminates almost--surely on $\eta$, if and only if the simulation of $Q$ terminates on every input.
Thus $Q \in \UHP$ if and only if $(P',\, \eta) \in \AST$.
%
%
\qed
\end{proof}
While for ordinary programs there is a complexity leap when moving from the halting problem for some given input to the universal halting problem, we establish that \emph{there is no such leap in the probabilistic setting}, i.e.\ $\UAST$ is as hard as $\AST$:
\begin{theorem}
\label{UASTisPi2comp}
$\UAST$ is $\Pi_2^0$--complete.
\end{theorem}
\begin{proof}
For showing $\UAST \in \Pi_2^0$, consider that by Theorem \ref{ASTisPi2comp} there must exist a decidable relation $\Problem R$ such that $(P,\, \eta) \in \AST$ iff $\forall \, y_1 \: \exists \, y_2\colon (y_1,\, y_2,\, P,\, \eta) \in \Problem R$.
By that we have that $P \in \UAST$ iff $\forall \, \eta\: \forall \, y_1 \: \exists \, y_2\colon (y_1,\, y_2,\, P,\, \eta) \in \Problem R$, which is a $\Pi_2^0$--formula.
It remains to show that $\UAST$ is $\Pi_2^0$--hard.
This can be done by proving $\AST \leqm \UAST$ as follows:
On input $(Q,\, \eta)$ the reduction function $f\colon \AST \leqm \UAST$ computes a probabilistic program $P$ that first initializes all variables according to $\eta$ and then executes $Q$.
\qed
\end{proof}
We now investigate the computational hardness of deciding \emph{positive} almost--sure termination:
It turns out that deciding $\PAST$ is---although still undecidable---computationally more benign than deciding $\AST$, namely it is $\Sigma_2^0$--complete.
Thus $\PAST$ becomes semi--decidable when given access to an $\HP$--oracle whereas $\AST$ does not.
We establish $\Sigma_2^0$--hardness by a reduction from $\cUHP$.
This result is particularly counterintuitive as it means that for each ordininary program that \emph{does not halt} on all inputs, we can \emph{compute} a probabilistic program that \emph{does halt} within an expected finite number of steps.
\begin{theorem}
\label{PASTisSigma2comp}
$\PAST$ is $\Sigma_2^0$--complete.
\end{theorem}
\begin{proof}
For $\PAST \in \Sigma_2^0$ consider the following:
\normalsize\begin{align*}
~&(P,\, \eta) \in \PAST\\
\Longleftrightarrow~&\infty > \Exp{P,\eta}{\downarrow}\\
\Longleftrightarrow~&\exists \, c \colon c > \Exp{P,\eta}{\downarrow}\\
\Longleftrightarrow~&\exists \, c \colon c > \sum_{k = 0}^{\infty} \left(1-\sum_{w \in \{L,\, R\}^{{}\leq k}} ~ \alpha\big(\T_{k}(\sigma_{P,\eta},\, w) \big)\right)\\
\Longleftrightarrow~&\exists \, c \: \forall\, \ell \colon c > \sum_{k = 0}^{\ell} \left(1-\sum_{w \in \{L,\, R\}^{{}\leq k}} ~ \alpha\big(\T_{k}(\sigma_{P,\eta},\, w) \big)\right)\\
\Longrightarrow~&\PAST \in \Sigma_2^0
\end{align*}
It remains to show that $\PAST$ is $\Sigma_2^0$--hard.
For that we use a reduction function $f\colon \cUHP \leqm \PAST$ with 
$f(Q) = (P,\, \eta)$, where $\eta$ is arbitrary but fixed and $P$ is the  program
\begin{alltt}
\(c\) := 1; \(i\) := 0; \(x\) := 0; term := 0; \(\mathit{InitQ}\);

while (\(c\) \(\neq\) 0)\{
    \(\mathit{StepQ}\); 
    if (term = 1)\{
        \(\mathit{Cheer}\); \(i\) := \(i\) + 1; term := 0; \(\mathit{InitQ}\)
    \};
    \{\(c\) := 0\} [0.5] \{\(c\) := 1\}; \(x\) := \(x\) + 1
\} \(\textnormal{,}\)
\end{alltt}
where $\mathit{InitQ} \in \sonprogs$ is a program that initializes a simulation of the program $Q$ on input $g_Q(i)$ (recall the bijection $g_Q\colon \bbbn \rightarrow \mathbb V$ from Theorem \ref{RisSigma2comp}), $\mathit{StepQ} \in \sonprogs$ is a program that does one single (further) step of that simulation and sets \texttt{term} to 1 if that step has led to termination of $Q$, and $\mathit{Cheer} \in \sonprogs$ is a program that executes $2^x$ many effectless steps.
In the following we refer to this as ``cheering"\footnote{The program $P$ cheers as it was able to prove the termantion of $Q$ on input $g_Q(i)$.}.

\textit{{Correctness of the reduction:}}
Intuitively, the program $P$ starts by simulating $Q$ on input $g_Q(0)$.
During the simulation, it---figuratively speaking---gradually looses interest in further simulating $Q$ by tossing a coin after each simulation step to decide whether to continue the simulation or not.
If eventually $P$ finds that $Q$ has halted on input $g_Q(0)$, it ``cheers" for a number of steps exponential in the number of coin tosses that were made so far, namely for $2^x$ steps.
$P$ then continues with the same procedure for the next input $g_Q(1)$, and so on.

The variable $x$ keeps track of the number of loop iterations (starting from 0), which equals the number of coin tosses. 
The $x$--th loop iteration takes place with probability $\nicefrac{1}{2^x}$.
One loop iteration consists of a constant number of steps $c_1$ in case $Q$ did not halt on input $g_Q(i)$ in the current simulation step.
Such an iteration therefore contributes $\nicefrac{c_1}{2^x}$ to the expected runtime of the probabilistic program $P$.
In case $Q$ did halt, a loop iteration takes a constant number of steps $c_2$ plus $2^x$ additional ``cheering" steps.
Such an iteration therefore contributes $\nicefrac{c_2 + 2^x}{2^x} = \nicefrac{c_2}{2^x} + 1 > 1$ to the expected runtime.
Overall, the expected runtime of the program $P$ roughly resembles a geometric series with exponentially decreasing summands.
However, for each time the program $Q$ halts on an input, a summand of the form $\nicefrac{c_2}{2^x} + 1$ appears in this series.
There are now two cases: 

(1) $Q \in \cUHP$, so there exists some input $\eta$ with minimal $i$ such that $g_Q(i) = \eta$ on which $Q$ does not terminate.
In that case, summands of the form $\nicefrac{c_2}{2^x} + 1$ appear only $i-1$ times in the series and therefore, the series converges---the expected time until termination is finite, so $(P,\, \eta) \in \PAST$.

(2) $Q \not\in \cUHP$, so $Q$ terminates on every input.
In that case, summands of the form $\nicefrac{c_2}{2^x} + 1$ appear infinitely often in the series and therefore, the series diverges---the expected time until termination is infinite, so $(P,\, \eta) \not\in \PAST$.
\qed
\end{proof}
The final problem we study is \emph{universal} positive almost--sure termination. 
In contrast to the non--positive version, we do have a complexity leap when moving from non--universal to universal positive almost--sure termination.
We will establish that $\UPAST$ is $\Pi_3^0$--complete and thus even harder to decide than $\UAST$.
For the reduction, we make use of the following $\Pi_3^0$--complete problem:
\begin{theorem}[The Cofiniteness Problem~\textnormal{\textbf{\cite{odifreddi2}}}]
\label{COFcomp}
The \textbf{cofiniteness problem} is a subset $\boldsymbol\COF \subset \sonprogs$, which is characterized as $P \in \COF$ iff $\big\{\eta ~\big|~ (P,\, \eta) \in \HP\big\}$ is cofinite.
Let $\boldsymbol\cCOF$ denote the \textbf{complement of $\boldsymbol\COF$}, i.e.\ $\cCOF = \sonprogs \setminus \COF$.
$\COF$ is $\Sigma_3^0$--complete and $\cCOF$ is $\Pi_3^0$--complete.
\end{theorem}
%
%
%
%
\begin{theorem}
\label{UPASTisPi3comp}
$\UPAST$ is $\Pi_3^0$--complete.
\end{theorem}
\begin{proof}
By Theorem \ref{PASTisSigma2comp}, there exists a decidable relation $\Problem R$, such that $(P,\, \eta) \in \PAST$ iff $\exists \, y_1 \: \forall\, y_2\colon (y_1,\, y_2,\, P,\, \eta) \in \Problem R$.
Therefore $\UPAST$ is definable by $P \in \UPAST$ iff $\forall\, \eta \: \exists \, y_1 \: \forall\, y_2\colon (y_1,\, y_2,\, P,\, \eta) \in \Problem R$ which gives a $\Pi_3^0$--formula.

It remains to show that $\UPAST$ is $\Pi_3^0$--hard.
For that we many--one reduce the $\Pi_3^0$--complete complement of the cofiniteness problem to $\UPAST$ using the following function $f\colon \cCOF \leqm \UPAST$: $f$ takes an ordinary program $Q$ as its input and computes the following probabilistic program $P$:
\begin{alltt}
\(c\) := 1; \(x\) := 0; term := 0; \(\mathit{InitQ}\);
while (\(c\) \(\neq\) 0)\{
    \(\mathit{StepQ}\); if (term = 1)\{\(\mathit{Cheer}\); \(i\) := \(i\) + 1; term := 0; \(\mathit{InitQ}\)\}
    \(x\) := \(x\) + 1; \{\(c\) := 0\} [0.5] \{\(c\) := 1\}
\},
\end{alltt}
where $\mathit{InitQ}$ is a program that initializes a simulation of the program $Q$ on input $g_Q(i)$ (recall the bijection $g_Q\colon \bbbn \rightarrow \mathbb V$ from Theorem \ref{RisSigma2comp}), $\mathit{StepQ}$ is a program that does one single (further) step of that simulation and sets \texttt{term} to 1 if that step has led to termination of $Q$, and $\mathit{Cheer}$ is a program that executes $2^x$ many effectless steps.

\textit{{Correctness of the reduction:}}
The only difference to the program from the proof of Theorem \ref{PASTisSigma2comp} is that the variable $i$ is not initialized with 0.
Thus, on input $\eta$ the program $P$ also simulates $Q$ succesively on all inputs but starting from input $g_Q\big(\eta(i)\big)$ instead of $g_Q(0)$.

The problem $\cCOF$ can alternatively be defined as $Q \in \cCOF$ iff $\big\{\eta ~|~ (Q,\, \eta) \in \cHP\}$ is infinite.
There are now two cases:

(1) $Q \in \cCOF$. Thus, there are infinitely many inputs on which $Q$ does not terminate.
So no matter with which number $\eta(i) \in \mathbb N$ the variable $i$ is initialized, the variable $i$ will eventually be incremented to some value $j$ such that $Q$ does not terminate on $g_Q(j)$ and therefore, in the while--loop the \texttt{if}--branch with the ``cheering" steps will eventually \emph{not} be executed anymore.
Consequently, the expected time until termination of $P$ on any input $\eta$ is finite and therefore $P \in \UPAST$.

(2) $Q \not\in \cCOF$. Then there are only finitely many inputs on which $Q$ does not terminate.
Say $j$ is minimal such that $Q$ does not terminate on input $g_Q(j)$, i.e.\ the program $Q$ terminates on all other inputs $g_Q(i')$ with $j < i'$.
So when running the program $P$ on some input $\eta$ with $\eta(i) > j$, in the while--loop the \texttt{if}--branch with the ``cheering" steps will be executed infinitely often.
Consequently, the expected time until termination of $P$ on that input $\eta$ is infinite and therefore $P \not\in \UPAST$.
%
\qed
\end{proof}

\section{Conclusion}
\begin{figure}[t!]
\begin{center}
\begin{tikzpicture}
\pgftransformscale{.947}
\draw[white, use as bounding box] (-4.85,.5) rectangle (4.85,7.5);

\draw[very thick] (-3,0) -- (3,0);
\draw[very thick] (-3,0) -- (-3,7);
\draw[very thick] (3,0) -- (3,7);
\draw[very thick, dotted] (-3, 7) -- (-3, 7.5);
\draw[very thick, dotted] (3, 7) -- (3, 7.5);


\node at (-2.575, 1.7) {$\boldsymbol{\Sigma_1^0}$};
\node at (2.575, 1.7) {$\boldsymbol{\Pi_1^0}$};
\node at (0, 1.2) {$\boldsymbol{\Delta_1^0}$};
\node at (-2.575, 3.7) {$\boldsymbol{\Sigma_2^0}$};
\node at (2.575, 3.7) {$\boldsymbol{\Pi_2^0}$};
\node at (0, 3.2) {$\boldsymbol{\Delta_2^0}$};
\node at (-2.575, 5.7) {$\boldsymbol{\Sigma_3^0}$};
\node at (2.575, 5.7) {$\boldsymbol{\Pi_3^0}$};
\node at (0, 5.2) {$\boldsymbol{\Delta_3^0}$};

\draw (-3,6) parabola bend (-2,6) (3,4) ;
\draw (-3,4) parabola bend (2,6) (3,6) ;

\draw (-3,4) parabola bend (-2,4) (3,2) ;
\draw (-3,2) parabola bend (2,4) (3,4) ;

\draw (-3,2) parabola bend (-2,2) (3,0) ;
\draw (-3,0) parabola bend (2,2) (3,2) ;

\node at (0, 6.5) {\large $\vdots$};

\node at (-2.5, .9) {\textcolor{gray}{\scalebox{.75}{$\HP$}}};
\node at (2.5, .9) {\textcolor{gray}{\scalebox{.75}{$\cHP$}}};
\node at (-2.5, 2.95) {\textcolor{gray}{\scalebox{.75}{$\cUHP$}}};
\node at (2.6, 2.85) {\textcolor{gray}{\scalebox{.75}{$\UHP$}}};
\node at (-2.5, 4.95) {\textcolor{gray}{\scalebox{.75}{$\COF$}}};
\node at (2.6, 4.85) {\textcolor{gray}{{\scalebox{.75}{$\cCOF$}}}};

{\node at (-1.6, 1.6) {\scriptsize{$\LEXP$}};}
{\node at (-4, 1) {\tiny \begin{tabular}{r}{semi--decidable}\end{tabular}};}
{\node at (0, 0.55) {\tiny \begin{tabular}{r}{decidable}\end{tabular}};}

{\node at (-1.3, 3.7) {\scriptsize{$\PAST$}};}
{\node at (-1.7, 3.35) {\scriptsize{$\REXP$}};}
{\node at (-4, 3) {\tiny \begin{tabular}{r}{with access to}\\{$\HP$--oracle:}\\{semi--decidable}\end{tabular}};}

{\node at (2.3, 3.2) {\scriptsize{$\EXP$}};}
{\node at (1.1, 3.7) {\scriptsize{$\AST$}};}
{\node at (4, 3) {\tiny \begin{tabular}{l}{not}\\ {semi--decidable;}\\{even with}\\{access to}\\{$\HP$--oracle}\end{tabular}};}
{\node at (4, 5) {\tiny \begin{tabular}{l}{not}\\ {semi--decidable;}\\{even with}\\{access to}\\{$\UHP$--oracle}\end{tabular}};}
{\node at (1.8, 3.5) {\scriptsize{$\UAST$}};}
{\node at (1.6, 5.6) {\scriptsize{$\UPAST$}};}

\end{tikzpicture}
\end{center}
\caption{The complexity landscape of determining expected outcomes and deciding (universal) (positive) almost--sure termination.}
\label{fighier}
\end{figure}
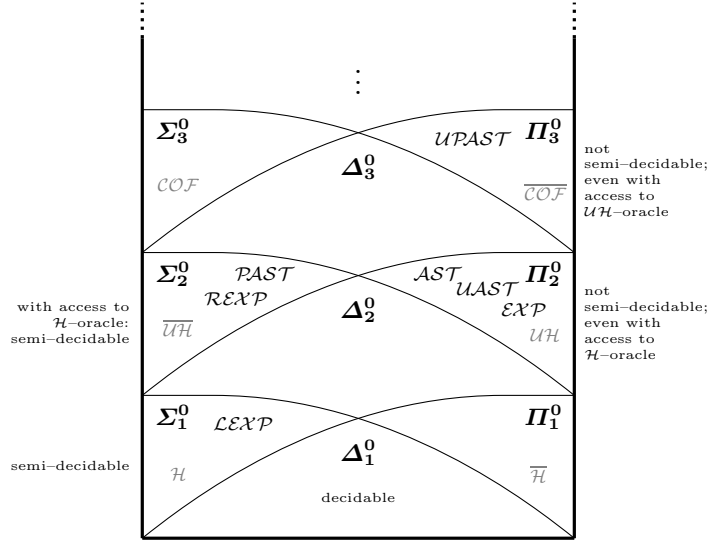
We have studied the computational complexity of solving a variety of natural problems which appear in the analysis of probabilistic programs: Computing lower bounds, upper bounds, and exact expected outcomes  ($\LEXP$, $\REXP$, and $\EXP$), deciding non--universal and universal almost--sure termination ($\AST$ and $\UAST$), and deciding non--universal and universal positive almost--sure termination ($\PAST$ and $\UPAST$).
Our complexity results are summarized in Figure~\ref{fighier}.
All examined problems are complete for their respective level of the arithmetical hierarchy.

Future work consists of identifying program subclasses for which some of the studied problems become easier.
One idea towards this would be to investigate the use of quantifier--elimination methods such as e.g.\ Skolemization.

\subsubsection*{Acknowledgements}
The authors would like to thank Luis Mar\'{i}a Ferrer Fioriti (Saarland University) and Federico Olmedo (RWTH Aachen) for the fruitful discussions on the topics of this paper.

\bibliographystyle{splncs}
\bibliography{literature}

\clearpage
\appendix

\end{document}